\documentclass[a4paper,UKenglish,cleveref, autoref]{lipics-v2019}

\usepackage{amsmath}
\usepackage{amssymb}
\usepackage{amsthm}
\usepackage{stmaryrd}

\usepackage{graphicx}
\usepackage{latexsym}
\usepackage{url}
\usepackage{color}
\usepackage{tikz}
\usetikzlibrary{automata}
\usepackage{algorithm,algorithmic}

\newcommand{\rrangle}{\rangle\!\rangle}
\newcommand{\llangle}{\langle\!\langle}
\newcommand{\eventually}{\Diamond}                                 
\newcommand{\strat}[1]{\llangle#1\rrangle}                         
\newcommand{\nostrat}[1]{\llbracket #1 \rrbracket}         

\newcommand{\word}[1]{\langle #1\rangle}                           
\newcommand{\set}[1]{\{#1\}}
\newcommand{\LT}[1]{\stackrel{#1}{\rightarrow}}

\newcommand{\Nat}{\mathbb{N}}
\newcommand{\A}{\texttt{Act}}

\newcommand\Prop{\texttt{Prop}}
\newcommand\Sim{\sqsubseteq}                                        

\newcommand\dist{\mathcal{D}}

\newcommand\step{\delta}                                         

\renewcommand\L{\mathcal{L}}

\newcommand\G{\mathcal{G}}                                       

\newcommand\lift[1]{\overline{#1}}                               
\newcommand\R{\mathcal{R}}                                       
\newcommand\Supp[1]{\lceil#1\rceil}                              
\newcommand\pd[1]{\widetilde{#1}}                                

\newcommand{\sem}[1]{\nostrat{#1}}                               
\newcommand\PATL{PATL}                                           
\newcommand\pone{{\cal I}}
\newcommand\ptwo{{\cal {I\!I}}}
\newcommand\commentout[1]{}

\newcommand{\power}{\mathcal{P}}

\newcommand{\LP}{\L^{\oplus}}
\newcommand{\LM}{\L^{\mu}}
\newcommand{\VAR}{\mathcal{V}}

\title{Characterising Probabilistic Alternating Simulation for Concurrent Games} 

\titlerunning{PA-simulation for concurrent games}

\author{Chenyi Zhang}{College of Information Science and Technology \& \\
College of Cyber Security\\
Jinan University, China}{chenyi\_zhang@jnu.edu.cn}{}{}

\author{Jun Pang}{Faculty of Science, Technology and Communication \& \\
Interdisciplinary Centre for Security, Reliability and Trust\\
University of Luxembourg, Luxembourg}{jun.pang@uni.lu}{}{}

\authorrunning{C.Zhang and J.Pang}

\Copyright{Chenyi Zhang and Jun Pang}

\ccsdesc[100]{Theory of computation}
\ccsdesc[100]{Modal and temporal logics}

\keywords{Concurrent game, Probabilistic alternating simulation, Logic characterisation}

\category{}

\nolinenumbers 

\EventEditors{AAA}
\EventNoEds{2}
\EventLongTitle{Conference}
\EventShortTitle{Submission}
\EventAcronym{BBB}
\EventYear{2019}
\EventDate{2019}
\EventLocation{XXX}
\EventLogo{}
\SeriesVolume{XX}
\ArticleNo{XX}

\begin{document}

\maketitle

\begin{abstract}
Probabilistic game structures combine both nondeterminism and stochasticity,
where players repeatedly take actions simultaneously to move to the next state of the concurrent game.
Probabilistic alternating simulation is an important tool
to compare the behaviour of different probabilistic game structures.
In this paper, we present a sound and complete modal characterisation of this simulation relation
by proposing a new logic based on probabilistic distributions.
The logic enables a player to enforce a property in the next state or distribution.
Its extension with fixpoints, which also characterises the simulation relation,
can express a lot of interesting properties in practical applications.
\end{abstract}

\section{Introduction}

Simulation relations and bisimulation relations~\cite{Mil89} are important
research topics in concurrency theory.
In the classical model of labelled transition systems (LTS),
simulation and bisimulation have been proved useful for comparing the behaviour of concurrent
systems. The modal characterisation problem has been studied both in classical and in probabilistic systems,
i.e., the Hennessy-Milner logic (HML)~\cite{HM85} that characterises image-finite LTS, and various
modal logics have been proposed to characterise strong and weak probabilistic (bi)simulation
in the model of probabilistic automata~\cite{PS07,DG10,HPSWZ11}.
To study multi-player games, the concurrent game structure (GS)~\cite{AHK02} is a model
that defines a system that evolves while interacting with outside \emph{players}.
As a player's behaviour is not fully specified within a system, GS are often also known as \emph{open} systems.
Alternating simulation (A-simulation) is defined in GS focusing on players ability to enforce temporal properties
specified in alternating-time temporal logic (ATL)~\cite{AHK02}, and A-simulation is shown to be
sound and complete to a fragment of ATL~\cite{AHKV98}.

In this paper, we work on the model of probabilistic game structure (PGS) which has probabilistic transitions.
PGS also allows probabilistic (or mixed) choices of players. In this setting, we assume all players have complete
information about system states.
The simulation relation in PGS, called probabilistic alternating simulation (PA-simulation),
has been shown to preserve a fragment of probabilistic alternating-time temporal logic (PATL)
under \emph{mixed strategies}~\cite{ZP10}. Given the classical results of modal characterisations
for (non-probabilistic) LTS, probabilistic automata, as well as for (non-probabilistic) game structures,
we investigate if similar correspondence exists for processes and modal logics in the domain of
concurrent games with probabilistic transitions and mixed strategies.  We find that
such a correspondence still holds by adapting a modal logic with nondeterministic distributions
extended from the work of~\cite{DG10}.
As a by-product, we extend that modal logic with fixpoint operators and study its
expressiveness power. Notably, similar to the fixpoint logics in~\cite{Liu2015,Song2019},
the least fixpoint modality in our logic only expresses finite reachability, a property in line with the
original $\mu$-calculus~\cite{Kozen83}, which somehow in the probabilistic setting may not be powerful enough
to express certain reachability properties that require infinite accumulation of moves in a play.

\smallskip\noindent
\textbf{Contributions.}
This paper studies modal characterisation of the probabilistic alternation simulation relation
in probabilistic concurrent game structures, which defines a novel modal logic based on probabilistic distributions. This new logic expresses
a player's power to enforce a property in the next state or distribution. The logic also incorporates both probabilistic and nondeterminstic
features that need to be considered during the two-player interplay. The second contribution is the introduction of a fixpoint logic,
which also characterises the simulation relation, extended from that modal logic. The expressive power of the logic has been illustrated
by examples.

\section{Preliminaries}
\label{sec:preliminaries}

A \emph{discrete probabilistic distribution} $\Delta$ over a set $S$ is a function of type $S\rightarrow[0,1]$,
where $\sum_{s\in S}\Delta(s)=1$.
We write $\dist(S)$ for the set of all such distributions, ranged over by symbols $\Delta,\Theta,\dots$.
Given a set $T\subseteq S$, $\Delta(T)=\sum_{s\in T}\Delta(s)$,
i.e., the probability for the given set $T$.
Given an index set $I$, a list of distributions $\word{\Delta_i}_{i\in I}$ and
a list of values $\word{p_i}_{i\in I}$ where $p_i\in[0,1]$ for all $i\in I$
and $\sum_{i\in I}p_i=1$, we have that $\sum_{i\in I}p_i\Delta_i$ is also a distribution.
If $|I|=2$ we may also write $\Delta_1\oplus_\alpha\Delta_2$ for the distribution
$p_1\Delta_1+p_2\Delta_2$ where $p_1=\alpha$ and $p_2=1-\alpha$.
For $s\in S$, $\pd{s}$ represents a \emph{point (or Dirac) distribution}
satisfying $\pd{s}(s)=1$ and $\pd{s}(t)=0$ for all $t\neq s$.
Given $\Delta\in\dist(S)$, we define $\Supp{\Delta}$ as the set $\set{s\in S\mid\Delta(s)>0}$,
which is the \emph{support}~of~$\Delta$.

We work on a model called probabilistic (concurrent) game structure (PGS)
with two players $\pone$ and $\ptwo$
(though we believe our results can be straightforwardly extended to handle a finite set of players
as in the standard concurrent game structures~\cite{AHK02}).
Each player has complete information about the PGS at any time during a play.
Let $\Prop$ be a finite set of propositions.

\begin{definition}\label{def:PGS}
A probabilistic game structure (PGS) $\G$ is a tuple $\word{S, s_0, L, \A,\step}$, where
\begin{itemize}
\item $S$ is a finite set of states, with $s_0$ the initial state;
\item $L:S\rightarrow 2^{\Prop}$ is the labelling function which
      assigns to each state $s\in S$ a set of propositions true in $s$;
\item $\A=\A_\pone\times\A_\ptwo$ is a finite set of joint actions, where $\A_\pone$ and $\A_\ptwo$ are, respectively,
      the sets of actions for players $\pone$ and $\ptwo$;
\item $\step:S\times\A\rightarrow \dist(S)$ is a transition function.
\end{itemize}
\end{definition}
If in state $s$ player $\pone$ performs action $a_1$ and player $\ptwo$ performs action $a_2$,
then $\step(s,\word{a_1,a_2})$ is the distribution for the next states.
During each step the players choose their next moves simultaneously.
\begin{figure}[!t]
\centering
\includegraphics[height=4.2cm]{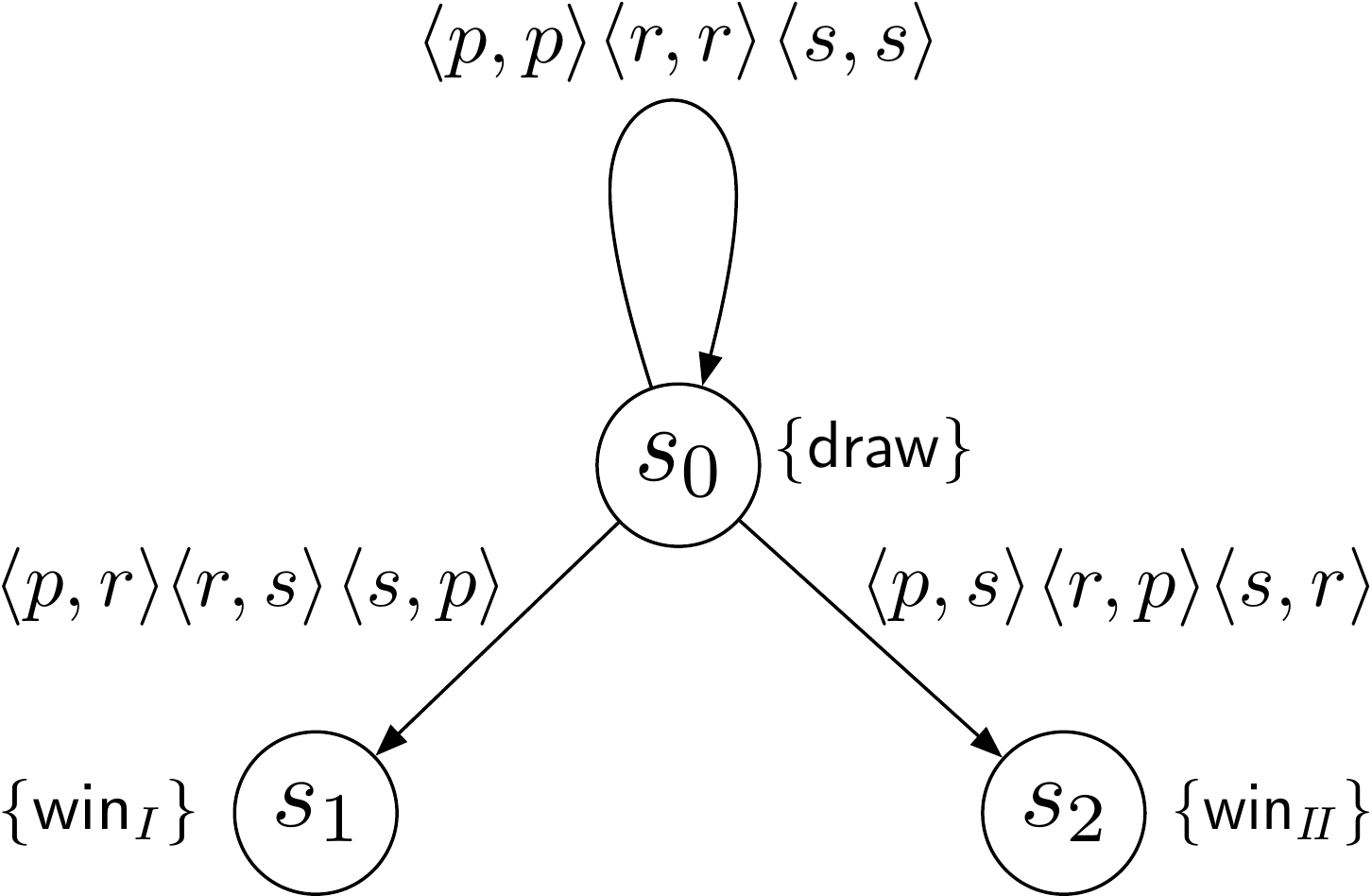}
\caption{The PGS for the repeated rock-paper-scissors game.}\label{fig:pgs}
\end{figure}

\begin{example}
Figure~\ref{fig:pgs} presents the PGS of two players repeatedly playing
the rock-paper-scissors game.\footnote{A similar example was used in~\cite{KNPS18}.
A concurrent stochastic structure (CSG), as defined in~\cite{KNPS18},
in a PGS, in the sense that all the probability distributions involved in the CSG
are point distributions.}
It has three states $s_0$, $s_1$, and $s_2$, with $s_0$ being the initial state.
Each state is labelled with an atomic proposition indicating the result of
a round of the game (which payer wins or there is a draw).
For instance, in state $s_1$ player $\pone$ wins the game.
Actions of the players are $r$ (representing playing rock),
$p$ (representing playing paper), $s$ (representing playing scissors).
The joint actions $\word{a_1,a_2}$ with $a_1, a_2\in \{r, p, s\}$
are depicted along with the transitions.
The function  $\step$ describes the transitions from state $s$ and state $s'$ as shown in Figure~\ref{fig:pgs}.
The winning states $s_1$ and $s_2$ are absorbing, i.e., all actions from there make self-transitions, and the game effectively terminates there.
\end{example}



We define a \emph{mixed action} of player $i$
as a function from states to distributions on $\A_i$, ranged over by $\pi,\pi_1,\sigma\dots$, and write $\Pi_i$ for the set of mixed actions from player $i$.
In particular, $\pd{a}$ is a \emph{deterministic} mixed action
which always chooses $a$ with probability $1$ in all states.
\begin{example}
In the rock-paper-scissors game (see Figure~\ref{fig:pgs}),
the mixed action with probability $\frac{1}{3}$ for each of the actions ($r$, $p$ and $s$)
is known as the optimal strategy for both players.
\end{example}

We generalise the transition function $\step$ by $\pd\step$ to handle mixed actions.
Given $\pi_1\in\Pi_\pone$ and $\pi_2\in\Pi_\ptwo$, for all $s, t\in S$,
we have $\pd{\step}(s,\word{\pi_1,\pi_2})(t)=\sum_{a_1\in\A_\pone,a_2\in\A_\ptwo}\pi_1(s)(a_1)\cdot\pi_2(s)(a_2)\cdot\step(s,\word{a_1,a_2})(t)$.
We further extend the function $\pd\step$ to handle transitions from distributions to distributions.
Formally, given a distribution $\Delta$, $\pi_1\in\Pi_\pone$ and $\pi_2\in\Pi_\ptwo$, for all $s\in S$, we have
$\pd{\step}(\Delta,\word{\pi_1,\pi_2})(s)=\sum_{t\in\Supp{\Delta}}\Delta(t)\cdot\pd{\step}(t,\word{\pi_1,\pi_2})(s)$.
For better readability, sometimes we write $\Delta\xrightarrow{\pi_1,\pi_2}\Theta$ if $\Theta=\pd\step(\Delta,\word{\pi_1,\pi_2})$.

Let $\leq\ \subseteq S\times S$ be a partial order, define $\leq_{Sm}\subseteq\power(S)\times\power(S)$,
by $P\leq_{Sm}Q$ if for all $t\in Q$ there exists $s\in P$ such that $s\leq t$.
In the literature this definition is known as the `Smyth order'~\cite{Smy78,Karger94} regarding `$\leq$'.

Relations in probabilistic systems usually require a notion of \emph{lifting}~\cite{JL91},
which extends the relations to the domain of distributions.%
\footnote{In a probabilistic system without explicit user interactions, state $s$ is simulated by state $t$
if for every $s\LT{a}\Delta_1$ there exists $t\LT{a}\Delta_2$ such that $\Delta_1$ is simulated by $\Delta_2$.}
Let $S$, $T$ be two sets and $\R\subseteq S\times T$ be a relation,
then $\lift\R\subseteq\dist(S)\times\dist(T)$ is a \emph{lifted relation} defined by $\Delta\,\lift\R\,\Theta$
if there exists a \emph{weight} function $w:S\times T\rightarrow[0,1]$ such that
\begin{itemize}
\item $\sum_{t\in T}w(s, t)=\Delta(s)$ for all $s\in S$,
\item $\sum_{s\in S}w(s, t)=\Theta(t)$ for all $t\in T$,
\item $s\R\, t$ for all $s\in S$ and $t\in T$ with $w(s,t)>0$.
\end{itemize}

The intuition behind the lifting is that each state in the support of one distribution may correspond
to a number of states in the support of the other distribution, and vice versa.
In the following section, we extend the notion of alternating simulation~\cite{AHKV98} to a probabilistic setting in the way of lifting.
The next example is taken from~\cite{Seg95} which shows how to lift a relation.

\begin{example}
In Figure~\ref{fig:lifting}, we have two sets of states $S=\set{s_1,s_2}$ and $T=\set{t_1,t_2,t_3}$,
and a relation $\R=\set{(s_1,t_1),(s_1,t_2),(s_2,t_2),(s_2,t_3)}$.
Suppose $\Delta(s_1)=\Delta(s_2)=\frac{1}{2}$ and
$\Theta(t_1)=\Theta(t_2)=\Theta(t_3)=\frac{1}{3}$, we may establish $\Delta\,\lift\R\,\Theta$. To check this, we define a weight function $w$ by:
$w(s_1,t_1)=\frac{1}{3}$, $w(s_1,t_2)=w(s_2,t_2)=\frac{1}{6}$, and $w(s_2,t_3)=\frac{1}{3}$.
The dotted lines in the graph indicate the allocation of weights
that is required to relate $\Delta$ to $\Theta$~via~$\lift\R$.
\end{example}

\begin{figure}[!t]
\centering
\includegraphics[height=4.0cm]{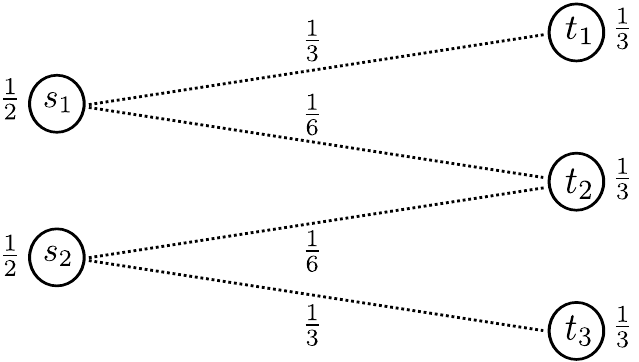}
\caption{An example showing how to lift one relation.}\label{fig:lifting}
\end{figure}

We present some properties of lifted relations.
First we show that, by combining distributions that are lift-related with the same weight on both sides,
we get the resulting distributions lift-related.

\begin{lemma}\label{lem:lift-comb}
Let $\R\subseteq S\times S'$ and $\word{p_i}_{i\in I}$ be a list of values satisfying
$\sum_{i\in I}p_i=1$, and $\Delta_i\,\lift\R\,\Delta'_i$ for $\Delta_i\in\dist(S)$ and $\Delta_i'\in\dist(S')$
for all $i$, then $\sum_{i\in I}p_i\Delta_i\,\lift\R\,\sum_{i\in I}p_i\Delta'_i$.
\end{lemma}

The following lemma states that, given two related distributions, if we split a distribution
on one side, then there exists a split on the other side to get a one-to-one matching with respect
to the lifted relation.

\begin{lemma}\label{lem:dist-split}
Let $\Delta\in\dist(S)$, $\Delta'\in\dist(S')$, $\R\subseteq S\times S'$, $\word{p_i}_{i\in I}$
be a list of values satisfying $\sum_{i\in I}p_i=1$.
If $\Delta\,\lift\R\,\Delta'$, then
\begin{enumerate}
\item for all lists of distributions $\word{\Delta_i}_{i\in I}$ with $\Delta_i\in\dist(S)$ for all $i\in I$, satisfying
      $\Delta=\sum_{i\in I}p_i\Delta_i$,
      there exists $\word{\Delta_i'}_{i\in I}$ with $\Delta_i'\in\dist(S')$ such that $\Delta'=\sum_{i\in I}p_i\Delta_i'$ and
      $\Delta_i\,\lift\R\,\Delta_i'$ for all $i\in I$;
\item for all lists of distributions $\word{\Delta_i'}_{i\in I}$ with $\Delta_i'\in\dist(S')$ for all $i\in I$, satisfying $\Delta'=\sum_{i\in I}p_i\Delta_i'$, there exists $\word{\Delta_i}_{i\in I}$ with $\Delta_i\in\dist(S)$ such that $\Delta=\sum_{i\in I}p_i\Delta_i$, and $\Delta_i\,\lift\R\,\Delta_i'$ for all $i\in I$.
\end{enumerate}
\end{lemma}

\section{Probabilistic Alternating Simulation}
\label{sec:simulation}

In concurrency models, simulation is used to relate states with respect to their behaviours.
For example, in a labelled transition system (LTS) $\word{S,A,\rightarrow}$, where $S$ is a set of states,
$A$ is a set of actions and $\rightarrow\ \subseteq S\times A\times S$ is the transition relation,
we say state $s$ is simulated by state $t$, written $s\leq t$,
if for every $s\LT{a}s'$ there exists $t\LT{a}t'$ such that $s'\leq t'$.
In this 
coinductive definition, state $t$ is able to simulate state $s$ by performing the same action $a$,
with their destination states still related.
Simulation is a useful tool in abstraction and refinement based verification,
as informally, in the above case, $t$ contains at least as much behaviour as $s$ does.
If the relation `$\leq$' is symmetric, then it is also a bisimulation.

In two-player non-probabilistic game structures (GS),
alternating simulation (A-simulation) is used to describe a player's ability to enforce certain temporal requirements
(regardless of the other player's behaviours)~\cite{AHKV98}.
In this paper we only focus on the ability of player $\pone$ in a two-player game.
Since in a game structure a transition requires the participation of both parties,
fixing player $\pone$'s input leaves a set of possible next states depending on player $\ptwo$'s inputs.
A-simulation is defined in the model of non-probabilistic game structures $\word{S,s_0,\A,L,\step}$,
which has a set of states $S$ with $s_0\in S$ the initial state,
$\A=\A_\pone\times\A_\ptwo$ the set of actions from players $\pone$ and $\ptwo$,
$L$ the labelling function and $\step:S\times\A\rightarrow S$ the transition function.
An A-simulation $\leq^A \subseteq S\times S$ is defined as follows.
Let $s,t\in S$, $s$ is A-simulated by $t$, i.e., $s\leq^A t$, if
\begin{itemize}
\item $L(s)=L(t)$, and
\item for all $a\in \A_\pone$ there exists $a'\in \A_\pone$ such that $\step(s,a)\leq_{Sm}^A\step(t,a')$,
where $\step(s,a)$ is the ``curried'' transition function defined by $\set{t\in S\mid \exists b\in \A_\ptwo:\step(s, \word{a,b})=t}$.
\end{itemize}

Intuitively, on state $t$ action $a'$ enforces a more restrictive set than action $a$ enforces on state $s$,
as shown by the Smyth-ordered relation $\leq^A_{Sm}$: for every $b'\in\A_\ptwo$ there exists $b\in\A_\ptwo$
such that $\step(s,\word{a,b})\leq^A\step(t,\word{a',b'})$.

Zhang and Pang extend A-simulation to probabilistic alternating simulation (PA-simulation) in PGS~\cite{ZP10}.
Their definition requires lifting of the simulation relation to derive a relation on distributions of states.

\begin{definition}\label{def:sim}
Given a PGS $\word{S, s_0, L, \A,\step}$, a \emph{probabilistic alternating simulation} (PA-simulation)
is a relation $\Sim\ \subseteq S\times S$ such that
$s\Sim t$ if
\begin{itemize}
\item $L(s)=L(t)$, and
\item for all $\pi_1\in\Pi_\pone$, there exists $\pi_2\in\Pi_\pone$,
such that $\pd\step(s, \pi_1)\mathrel{\lift{\Sim}}_{Sm}\pd\step(t, \pi_2)$,
where $\pd\step(s,\pi)=\set{\Delta\in\dist(S)\mid\exists\pi'\in\Pi_\ptwo:\pd\step(s, \word{\pi,\pi'})=\Delta}$.
\end{itemize}
\end{definition}
\noindent
If state $s$ PA-simulates state $t$ and $t$ PA-simulates $s$, we say $s$ and $t$ are
\emph{PA-simulation equivalent}, which is written $s\simeq t$.

\section{A Modal Logic for Probabilistic GS}
\label{sec:logic}
In the literature different modal logics have been introduced to
characterise process semantics at different levels.
Hennessy-Milner logic (HML)~\cite{HM85} provides a classical example
that has been proved to be equivalent to bisimulation semantics in image-finite LTS.
In other words, two states (or processes) satisfy the same set of HML formulas iff they are bisimilar.
For a more comprehensive survey we refer to~\cite{Gla01}.
\commentout{
For probabilistic systems, there are modal logics proposed and
proved to characterise strong and weak probabilistic (bi)simulation in the model of probabilistic automata~\cite{PS07,DG10,HPSWZ11}.}

In this section we propose a modal logic for PGS that characterises a player's abilitiy to enforce temporal properties.
We define a new logic $\LP$ in the spirit of the logic of Deng et al.~\cite{DGHMZ07,DG10}.
The syntax of the logic $\LP$ is presented below.

\[\varphi ::= p\mid \neg p\mid \bigwedge_{i\in I}\varphi_i\mid \bigvee_{i\in I}\varphi_i\mid \strat{\pone}\varphi \mid \bigoplus_{j\in J}p_j\varphi_j \mid \bigoplus_{j\in J}\varphi_j\]

In particular, $p$ is an atomic formula that belongs to the set $\Prop$.
Formula $\bigwedge_{i\in I}\varphi_i$ produces a 
conjunction, and $\bigvee_{i\in I}\varphi_i$ produces a 
disjunction, both via an index set $I$.
We then derive $\top=\bigwedge_{i\in\emptyset}\varphi_i$ is a formula that is true everywhere,
and $\bot=\bigvee_{i\in\emptyset}\varphi_i$ is a formula false everywhere.
$\strat{\pone}\varphi$  specifies  player $\pone$'s ability to enforce $\varphi$ in the next step.
The probabilistic summation operator $\bigoplus_{j\in J}p_j\varphi_j$ explicitly specifies
that a distribution satisfying such a formula should be split with pre-defined weights,
each part with weight $p_j$ satisfying sub-formula $\varphi_j$.
For a summation formula with index $J$, we may explicitly write down each component coupled by its weight,
such as in the way of $[p_1,\varphi_1]\oplus[p_2,\varphi_2]\oplus\ldots\oplus[p_{|J|},\varphi_{|J|}]$.
The operator $\bigoplus_{j\in J}\varphi_j$ allows arbitrary linear interpolation among formulas $\varphi_j$.
We only allow negation of formulas on the propositional level.
We use $\LP$ to denote the set of modal formulas defined by the above syntax.

The semantics of $\LP$ is presented as follows.
The interpretation of each formula is defined as a set of distributions of states in a finite PGS $\G=\word{S, s_0, L, \A,\step}$.

\begin{itemize}
\item $\sem{p}=\set{\Delta\in\dist(S)\mid \forall s\in\Supp{\Delta}: p\in L(s)}$;
\item $\sem{\neg p}=\set{\Delta\in\dist(S)\mid \forall s\in\Supp{\Delta}: p\not\in L(s)}$;
\item $\sem{\bigwedge_{i\in I}\varphi_i}=\bigcap_{i\in I}\sem{\varphi_i}$; $\sem{\bigvee_{i\in I}\varphi_i}=\bigcup_{i\in I}\sem{\varphi_i}$;
\item $\sem{\strat{\pone}\varphi} = \set{\Delta\in\dist(S)\mid \exists \pi_1\in\Pi_\pone: \forall \pi_2\in\Pi_\ptwo: \Delta\xrightarrow{\pi_1,\pi_2}\Theta\mbox{ implies }\Theta\in\sem{\varphi}}$;
\item $\sem{\bigoplus_{j\in J}p_j\varphi_j}=\set{\Delta\in\dist(S)\mid\Delta=\sum_{j\in J}p_j\Delta_j\wedge \forall j\in J: \Delta_j\in\sem{\varphi_j}}$;
\item $\sem{\bigoplus_{j\in J}\varphi_j}=\set{\Delta\in\dist(S)\mid\exists\set{p_j}_{j\in J}: \sum_{j\in J}p_j=1\wedge\Delta=\sum_{j\in J}p_j\Delta_j\wedge \forall j\in J: \Delta_j\in\sem{\varphi_j}}$;
\end{itemize}
Note here we say a distribution $\Delta$ satisfies a propositional formula if the formula holds in every state in the support of $\Delta$. The rest of the semantics is self-explained.

\smallskip\noindent
\textbf{Remark.}
The probabilistic modal logic proposed by Parma and Segala~\cite{PS07}
and Hermanns et al.~\cite{HPSWZ11} uses a fragment operator $[\varphi]_\alpha$,
such that a distribution $\Delta\models[\varphi]_\alpha$ iff there exists $\Delta_1,\Delta_2\in\dist(S)$
such that $\Delta=\Delta_1\oplus_\alpha\Delta_2$ and $\Delta_1\models\varphi$.
Informally, it states that a fragment of $\Delta$ with weight at least $\alpha$ satisfies $\varphi$.
Note that the summation operator of $\LP$ can be used to encode the fragment operator $[\varphi]_\alpha$,
in the way that $\Delta\models[\varphi]_\alpha$ iff $\Delta\models[\alpha,\varphi]\oplus[1-\alpha,\top]$.
Therefore, a straightforward adaptation of the logic by Parma and Segala~\cite{PS07} and Hermanns et al.~\cite{HPSWZ11} does not yield a more expressive logic than $\LP$.

The semantics of $\bigoplus_{j\in J}\varphi_j$ allows arbitrary linear interpolation among formulas $\varphi_j$.
Similar to the way treating probabilistic summations, one may write down $\bigoplus_{j\in J}\varphi_j$ by $[\varphi_1]\oplus[\varphi_2]\oplus\ldots\oplus[\varphi_{|J|}]$.
The following lemma is straightforward.
\begin{lemma}\label{lem:nondet-prob}
Let $\varphi=\bigoplus_{j\in J}p_j\varphi_j$ and $\varphi'=\bigoplus_{j\in J}\varphi_j$, and $\Delta\in\dist(S)$. We have $\Delta\models\varphi$ implies $\Delta\models\varphi'$.
\end{lemma}


%
%


Similar to most of the literature, given a PGS $\G$,
we define preorders on the set of states in $\G$ with respect to satisfaction of the modal logic $\LP$.

\begin{definition}\label{def:logic-order}
Given states $s,t\in S$, $s\Sim_{\LP}t$ if for all $\varphi\in\LP$, $s\models\varphi$ implies $t\models\varphi$.
If $s\Sim_{\LP}t$ and $t\Sim_{\LP}s$, we write $s\simeq_{\LP}t$.
\end{definition}


Now we state the first main result of the paper, and we leave its proof in the following two subsections.

\begin{theorem}\label{thm:modal-char}
Let $\G=\word{S, s_0, L, \A,\step}$ be a PGS, then for all $s,t\in S$, $s\Sim t$ iff $s\Sim_{\LP}t$.
\end{theorem}
\begin{proof}
By combining Theorem~\ref{thm:soundness} and Theorem~\ref{thm:completeness} and by treating $s$ and $t$ as point distributions $\pd{s}$ and $\pd{t}$.
\end{proof}

\begin{corollary}\label{cor:modal-char}
Let $\G=\word{S, s_0, L, \A,\step}$ be a PGS, then for all $s,t\in S$, $s\simeq t$ iff $s\simeq_{\LP}t$.
\end{corollary}

\subsection{A Soundness Proof}
\label{sec:soundness}

Since the notion of PA-simulation given in Definition~\ref{def:sim} is defined as a relation on states,
in the following we show that the lifted PA-simulation is also a \emph{simulation} on distributions
over the states, which is used as a stepping stone to our soundness result.
Similar to the way of treating distributions,
we also allow linear combination of mixed actions.

\begin{definition}\label{def:action-combine}
Given a list of mixed actions $\word{\pi_i}_{i\in I}$ (of player $\pone$), and $\word{p_i}_{i\in I}$ satisfying $\sum_{i\in I}p_i=1$, $\sum_{i\in I}p_i\pi_i$ is a mixed action defined by $\left(\sum_{i\in I}p_i\pi_i\right)(s)(a)=\sum_{i\in I}p_i\cdot\left(\pi_i(s)(a)\right)$
for all $s\in S$ and $a\in\A_\pone$.
\end{definition}

\begin{lemma}\label{lem:split:mix-action}
Let $s\in S$, $\pi\in\Pi_\pone$ and $\sigma=\sum_{i\in I}p_i\sigma_i\in\Pi_\ptwo$,
then $\pd\step(s,\word{\pi,\sigma})=\sum_{i\in I}p_i\cdot\pd\step(s,\word{\pi,\sigma_i})$.
\end{lemma}
\begin{lemma}\label{lem:split:mix-action-2}
Let $s\in S$, $\pi=\sum_{i\in I}p_i\pi_i\in\Pi_\pone$ and $\sigma\in\Pi_\ptwo$,
then $\pd\step(s,\word{\pi,\sigma})=\sum_{i\in I}p_i\cdot\pd\step(s,\word{\pi_i,\sigma}$.
\end{lemma}
\commentout{
\begin{proof}
Let $t\in S$, then
\[
\begin{array}{lllr}
     && \pd\step(s, \word{\pi,\sigma})(t) \\
    &=& \sum_{a_1\in\A_1}\sum_{a_2\in\A_2}\pi(s)(a_1)\cdot\sigma(s)(a_2)\cdot\step(s,\word{a_1,a_2})(t) \\
    &=& \sum_{a_1\in\A_1}\sum_{a_2\in\A_2}\pi(s)(a_1)\cdot\sum_{i\in I}p_i\cdot\sigma_i(s)(a_2)\cdot\step(s,a_1,a_2)(t) \\
    &=& \sum_{i\in I}p_i\cdot\left(\sum_{a_1\in\A_1}\sum_{a_2\in\A_2}\pi(s)(a_1)\cdot\sigma_i(s)(a_2)\cdot\step(s,a_1,a_2)(t)\right) \\
    &=& \sum_{i\in I}p_i\cdot\pd\step(s,\word{\pi,\sigma_i})(t)
\end{array}
\]
\end{proof}
The proof of the following lemma is similar to the above.
}
Lemma~\ref{lem:split:mix-action} and Lemma~\ref{lem:split:mix-action-2} show that we can distribute the 
distributions over actions out of a transition operator.

\begin{lemma}\label{lem:split:state}
Let $\Delta,\Theta\in\dist(S)$ with $\Delta=\sum_{i\in I}p_i\Delta_i$ and
$\Theta=\sum_{i\in I}p_i\Theta_i$, $\pi_1,\pi_2\in\Pi_\pone$ and $\sigma_1,\sigma_2\in\Pi_\ptwo$.
If $\pd\step(\Delta_i,\word{\pi_1,\sigma_1})\ \lift\Sim\ \pd\step(\Theta_i,\word{\pi_2,\sigma_2})$ for all $i\in I$,
then $\pd\step(\Delta,\word{\pi_1,\sigma_1})\ \lift\Sim\ \pd\step(\Theta,\word{\pi_2,\sigma_2})$.
\end{lemma}
Lemma~\ref{lem:split:state} allows to merge the simulation by component distributions on both sides of the relation.
%
%
%
%
The next auxiliary lemma states that given a PA-simulation on states, the lifted PA-simulation on distributions of states can be treated
as a \emph{simulation} via mixed actions of player $\pone$ and player $\ptwo$.

\begin{lemma}\label{lem:lifted:sim}
Let $\G=\word{S, s_0, \L, \A,\step}$ be a PGS, and $\Sim$ be a PA-simulation relation for $\G$. Given $\Delta\lift\Sim\Theta$,
for all player $\pone$ mixed actions $\pi_1$, there exists a player $\pone$ mixed action $\pi_2$,
such that $\pd\step(\Delta,\pi_1) \lift\Sim_{Sm} \pd\step(\Theta,\pi_2)$.
\end{lemma}
This can be proved by splitting distributions on both sides,
and then 
merge related components to form distributions on both sides of the lifted relation,
applying previous lemmas.
%
%

\begin{theorem}\label{thm:soundness}
Given $\varphi\in\LP$ and $\Delta\lift\Sim\Theta$, then $\Delta\models\varphi$ implies $\Theta\models\varphi$.
\end{theorem}
The theorem can be proved by structural induction on $\varphi$.
The base cases when $\varphi=p$ and $\neg p$ are straightforward.
For the \textsc{Induction step}, 
we only show the case of $\strat{\pone}\psi$, 
as the other cases are straightforward.
If $\varphi=\strat{\pone}\psi$, then there exists a player $\pone$ mixed actions $\pi_1$ such that for all
      player $\ptwo$ mixed actions $\sigma_1$, $\Delta\xrightarrow{\pi_1,\sigma_1}\Delta'$ and $\Delta'\models\psi$.
      By Lemma~\ref{lem:lifted:sim}, there exists a player $\pone$ mixed action $\pi_2$  such that
      $\pd\step(\Delta,\pi_1)\,\lift\Sim_{Sm}\pd\step(\Theta,\pi_2)$. Therefore, for all player $\ptwo$ mixed actions
      $\sigma_2$ there exists a player $\ptwo$ strategy $\sigma_1'$, such that $\Delta\xrightarrow{\pi_1,\sigma_1'}\Delta''$,
      $\Theta\xrightarrow{\pi_2,\sigma_2}\Theta'$, and $\Delta''\lift\Sim\Theta'$. Since $\Delta''\models\psi$,
      by I.H., $\Theta'\models\psi$. This shows that $\pi_2$ is the player $\pone$ mixed action for $\Theta\models\strat{\pone}\psi$.


\subsection{A Completeness Proof}
\label{sec:completeness}
The completeness is proved by approximating the relations $\Sim$ and $\Sim_{\LP}$.
For PA-simulation we construct relations $\Sim^n$ for $n\in\Nat$,
where $n$ denotes the numbers of steps that are required to check for a state to simulate another.
(Intuitively, the more steps to check, the harder for a pair of states to satisfy the relation.)
Similarly we define $\Sim^{\L}_n$, restricting to formulas in $\LP$ with \emph{size} up to $n$.
Then we prove that the relation $\Sim^{\L}_n$ is contained in $\Sim^n$ for all $n\in\Nat$.

\begin{definition}\label{def:formula-n}
Let $\LP_0$ be the set of formulas constructed by using only $p$, $\neg p$ and $\bigwedge_{i\in I}\varphi_i$.
For $n\in\Nat$, a formula $\varphi\in\LP_{n+1}$ if either $\varphi\in\LP_n$ or $\varphi$ is a conjunction of 
formulas of the form $\strat{\pone}\bigoplus_{i\in I}\sum_{j\in J}p_j\varphi_{i,j}$, where each $\varphi_{i,j}\in\LP_n$.
\end{definition}
Intuitively, formulas in $\LP_n$ require $n$ steps of transitions (for player~$\pone$) to enforce.
Given states $s,t\in S$, we write $s\Sim^\L_n t$, if for all $\varphi\in\LP_n$, $\pd{s}\models\varphi$ implies $\pd{t}\models\varphi$.
Similarly we define approximating relations for PA-simulation.

\begin{definition}\label{def:simulation-n}
Given $s,t\in S$, $s\Sim_0t$ if $L(s)= L(t)$.
For $n\in\Nat$, $s\Sim_{n+1}t$ if $s\Sim_nt$, and for all player $\pone$ mixed actions $\pi_1$ there exists a player $\pone$ mixed action $\pi_2$,
such that $\pd\step(s,\pi_1)\ (\lift\Sim_n)_{Sm} \pd\step(t,\pi_2)$.
\end{definition}

Before starting the completeness proof, we define formulas that characterise properties of the game states.
Let $s\in S$, the $0$-characteristic formula for
$s$ is $\phi_s^0=\bigwedge\set{p\mid p\in L(s)}\wedge\bigwedge\set{\neg q\mid
q\in \Prop\setminus L(s)}$. Plainly, the level $0$-characterisation considers
only propositional formulas. For a distribution, we specify the characteristic
formulas for the states in its support proportional to weights. The $0$-characteristic formula $\phi^0_\Delta$ for
distribution $\Delta$ is $\sum_{t\in\Supp{\Delta}}\Delta(t)\cdot\phi_t^0$. Given
all $n$-characteristic formulas defined, the ($n+1$)-characteristic formula
$\phi_s^{n+1}$ for state $s$ is $\bigwedge_{\pi\in\dist(\A_\pone)}\strat{\pone}\bigoplus_{\A_\ptwo}\phi^n_{\Delta_{\pi,b}}$,
where $\pd{s}\xrightarrow{\pi,\pd{b}}\Delta_{\pi,b}$. Similarly, an $n$-characteristic formula
$\phi^{n+1}_{\Delta}$ for distribution $\Delta$ is $\sum_{t\in\Supp{\Delta}}\Delta(t)\cdot\phi_{t}^{n+1}$.

Obviously every state or distribution satisfies its own characteristic formula, and
the following lemma can be proved straightforwardly by induction on $n\in\Nat$.

\begin{lemma}\label{lem:sound:characteristic}
For all $\Delta\in\dist(S)$, $\Delta\models\phi^n_\Delta$ for all $n\in\Nat$.
\end{lemma}

\begin{lemma}\label{lem:comp:approx}
For all states $s,t\in S$ and $n\in\Nat$, $s\Sim^\L_nt$ implies
$s\Sim_n t$.
\end{lemma}
\begin{proof}
For each $n\in\Nat$, we have $\pd{s}\models\phi^n_s$ by Lemma~\ref{lem:sound:characteristic}.
Let  $s\Sim^\L_n t$, then $\pd{t}\models\phi^n_s$.
We proceed by induction on the level of approximation $n$ to show that $s\Sim_n t$.

First we show that the state-based relation can be naturally carried over to distributions.
Suppose for all $s,t\in S$, $s\Sim^\L_n t$ implies
$\pd{t}\models\phi^n_s$. Given two distributions $\Delta,\Theta\in\dist(S)$ and let
$\Delta\ \lift\Sim^\L_n\ \Theta$. Then there exists a weight function $w$, such that
$\Delta=\sum_{s\in\Supp{\Delta},t\in\Supp{\Theta}:w(s,t)>0}w(s,t)\cdot\pd{s}$, and
$\Theta=\sum_{s\in\Supp{\Delta},t\in\Supp{\Theta}:w(s,t)>0}w(s,t)\cdot\pd{t}$, and
$s\Sim^\L_n t$ for all $w(s,t)>0$. Since $\phi^n_{\Delta}$ can be written as
$\sum_{s\in\Supp{\Delta},t\in\Supp{\Theta}:w(s,t)>0}w(s,t)\cdot\phi^n_s$, we must have
$\Theta\models\phi^n_{\Delta}$ as well.

\vspace{2mm}\noindent
\textsc{Base case:} Trivial.

\vspace{2mm}\noindent
\textsc{Induction step:} Let $\pd{t}\models\phi^{n+1}_s$, where
$\phi^{n+1}_s=\bigwedge_{\pi\in\dist(\A_\pone)}\strat{\pone}\bigoplus_{\A_\ptwo}\phi^n_{\Delta_{\pi,b}}$.
Then for
each $\pi\in\dist(\A_\pone)$, $\pd{t}\models\strat{\pone}\bigoplus_{b\in
\A_\ptwo}\phi^n_{\Delta_{\pi,b}}$. By definition there exists a player $\pone$
mixed action $\pi'$, such that for every player $\ptwo$ mixed action $\sigma$, we have
$\pd{t}\xrightarrow{\pi',\sigma}\Theta$ and $\Theta\models\bigoplus_{b\in
\A_\ptwo}\phi^n_{\Delta_{\pi,b}}$.
We need to show that $\pd\step(s, \pi)\mathrel{(\lift{\Sim}^\L_n)}_{Sm}\pd\step(t, \pi')$.

It suffices to check each $b\in\A_\ptwo$ from $t$ can be followed
by a player $\ptwo$ mixed action from $s$ to establish such a simulation.
Let $b'$ be a player $\ptwo$ action, and $\pd{t}\xrightarrow{\pi',\pd{b'}}\Theta$.
Since $\Theta\models\bigoplus_{b\in \A_\ptwo}\phi^n_{\Delta_{\pi,b}}$, there
exists a list of probability values $\word{p_{c}}_{c\in\A_\ptwo}$, such that
$\Theta\models\sum_{c\in\A_\ptwo}p_c\phi^n_{\Delta_{\pi,c}}$. Then by definition, we have
$\Theta=\sum_{c\in\A_\ptwo}p_c\cdot\Theta_c$, $\sum_{c\in\A_\ptwo}p_c = 1$ and
$\Theta_c\models\phi^n_{\Delta_{\pi,c}}$ for all $c\in\A_\ptwo$.
In state $s$, we define a player $\ptwo$ mixed action $\sigma$ satisfying $\sigma(s)(c)=p_c$ for
all $c\in\A_\ptwo$. Then by Lemma~\ref{lem:split:mix-action}, we have
$\pd\step(s,\word{\pi,\sigma})=\sum_{c\in\A_\ptwo}p_c\cdot\step(s,\word{\pi,c})=\Delta_{\pi,\sigma}$ for all $c\in\A_\ptwo$.
By Lemma~\ref{lem:lift-comb}, it suffices to show $\step(s,\word{\pi,c})\lift\Sim^n\Theta_c$ for all $c\in\A_\ptwo$.
%
Since $\Theta_c\models\phi^n_{\Delta_{\pi,c}}$, we have $\Delta_{\pi,c}\lift\Sim^n\Theta_c$ by I.H..
\end{proof}

Intuitively, by fixing a mixed strategy from player $\pone$, a transition in the PGS is bounded by deterministic actions
from player $\ptwo$, as mimicked in the structure of the characteristic formulas. The way of showing satisfaction of a characteristic
formula thus mimic the PA-simulation in the proof of Lemma~\ref{lem:comp:approx}.

\begin{theorem}\label{thm:completeness}
For all $s,t\in S$, $s\Sim^\L t$ implies $s\Sim t$.
\end{theorem}
\begin{proof}
In a finite state PGS (i.e., the space $S\times S$ is finite) there exists $n\in\Nat$ such that $\Sim\ =\ \Sim^n$.
Since $\Sim^\L\ \subseteq\ \Sim^\L_n$ for all $n$, and $\Sim^\L_n\ \subseteq\ \Sim_n$
by Lemma~\ref{lem:comp:approx}, we have $\Sim^\L\ \subseteq\ \Sim^n\ =\ \Sim$.
\end{proof}

\section{Probabilistic Alternating-time Mu-Calculus}
\label{sec:pamu}
Modal logics of finite modality depth are not enough to express temporal requirements such as 
``something bad never happens''. In this section,
we extend the logic $\LP$ into a Probabilistic Alternating-time $\mu$-calculus (PAMu), 
by adding variables and fixpoint operators.

\medskip


\begin{tabular}{rl}
 $\varphi$ $::=$ & $p\mid \neg p\mid \bigwedge_{i\in I}\varphi_i\mid \bigvee_{i\in I}\varphi_i\mid \strat{\pone}\varphi \mid \bigoplus_{j\in J}\varphi_j $  \\[5pt]
  &  $\mid \bigoplus_{j\in J}p_j\varphi_j \mid Z \mid \mu Z.\varphi \mid \nu Z.\varphi$  \\
 \end{tabular}

\medskip

Let the environment $\rho:\VAR\rightarrow\power(\dist(S))$ be a mapping from variables in $\VAR$ to sets of distributions on states,
and 
the semantics of the fixpoint operators of PAMu are defined in the standard way.
\begin{itemize}
\item $\sem{p}_\rho=\set{\Delta\in\dist(S)\mid \forall s\in\Supp{\Delta}: p\in L(s)}$;
\item $\sem{\neg p}_\rho=\set{\Delta\in\dist(S)\mid \forall s\in\Supp{\Delta}: p\not\in L(s)}$;
\item $\sem{\bigwedge_{i\in I}\varphi_i}_\rho=\bigcap_{i\in I}\sem{\varphi_i}_\rho$; $\sem{\bigvee_{i\in I}\varphi_i}=\bigcup_{i\in I}\sem{\varphi_i}_\rho$;
\item $\sem{\strat{\pone}\varphi}_\rho = \set{\Delta\in\dist(S)\mid \exists \pi_1\in\Pi_\pone: \forall \pi_2\in\Pi_\ptwo: \Delta\xrightarrow{\pi_1,\pi_2}\Theta\mbox{ implies }\Theta\in\sem{\varphi}_\rho}$;
\item $\sem{\bigoplus_{j\in J}p_j\varphi_j}_\rho=\set{\Delta\in\dist(S)\mid\Delta=\sum_{j\in J}p_j\Delta_j\wedge \forall j\in J: \Delta_j\in\sem{\varphi_j}_\rho}$;
\item $\sem{\bigoplus_{j\in J}\varphi_j}_\rho=\set{\Delta\in\dist(S)\mid\exists\set{p_j}_{j\in J}: \sum_{j\in J}p_j=1\wedge\Delta=\sum_{j\in J}p_j\Delta_j\wedge \forall j\in J: \Delta_j\in\sem{\varphi_j}_\rho}$;
\item $\sem{Z}_\rho = \rho(Z)$;
\item $\sem{\mu Z.\varphi}_\rho = \bigcap\set{D\subseteq\dist(S)\mid \sem{\varphi}_\rho[Z\mapsto D]\subseteq D}$;
\item $\sem{\nu Z.\varphi}_\rho = \bigcup\set{D\subseteq\dist(S)\mid D\subseteq \sem{\varphi}_\rho[Z\mapsto D]}$.
\end{itemize}

The set of closed PAMu formulas are the formulas with all variables bounded, which form the set $\LM$, and we can safely drop the environment $\rho$ for those formulas.
\begin{example}
For the rock-paper-scissors game in Figure~\ref{fig:pgs},
the property describing that player $\pone$ 
has a strategy to eventually win the game once
can be expressed as
$\mu Z. {\sf win}_{\pone} \vee\strat{\pone}Z$.
This property does not hold. However, player $\pone$ has a strategy to eventually win the game with probability almost $\frac{1}{2}$,
i.e., the system satisfies $\mu Z. ([\frac{1}{2}-\epsilon,{\sf win}_{\pone}]\oplus[\frac{1}{2}+\epsilon,\top]) \vee\strat{\pone}Z$ for arbitrarily small $\epsilon>0$.
We explain the reason why players can only enforce $\epsilon$-optimal strategies in later part of the section.
\end{example}

The logic characterisation of PA-Simulation can be extended to PAMu.
\begin{theorem}\label{thm:pamu}
Given $\Delta,\Theta\in\dist(S)$, $\Delta\lift\Sim\Theta$ iff
$\set{\varphi\in\LM\mid\Delta\models\varphi}\subseteq\set{\varphi\in\LM\mid\Theta\models\varphi}$.
\end{theorem}

Since $\LP$ is syntactically a sublogic of $\LM$, we only need to show the soundness of PA-simulation to the logic $\LM$.
We apply the classical approach of approximants for Modal Mu-Calculus~\cite{Bradfield06}.
Given formulas $\mu Z.\varphi$ and $\nu Z.\phi$, we define the following.

\smallskip

\begin{tabular}{ll}
 $\mu^0Z.\varphi = \bot$ &  $\nu^0Z.\phi = \top$ \\[5pt]
 $\mu^{i+1}Z.\varphi = \varphi[Z\mapsto\mu^{i}Z.\varphi]$\hspace{20pt} &  $\nu^{i+1}Z.\phi = \phi[Z\mapsto\nu^{i}Z.\phi]$  \\[5pt]
 $\mu^\omega Z.\varphi = \bigvee_{i\in\Nat}\mu^iZ.\varphi$ & $\nu^\omega Z.\phi = \bigwedge_{i\in\Nat}\nu^iZ.\phi$. \\
 \end{tabular}

\smallskip

Next we show approximants are semantically equivalent to the fixpoint formulas. 
\begin{lemma}\label{lem:approximants}
(1) $\sem{\mu^\omega Z.\varphi} = \sem{\mu Z.\varphi}$;\ (2) $\sem{\nu^\omega Z.\phi} = \sem{\nu Z.\phi}$.
\end{lemma}
We briefly sketch a proof of Lemma~\ref{lem:approximants}(1), and the proof for the other part of the lemma is similar.
To show $\sem{\mu^\omega Z.\varphi} \subseteq \sem{\mu Z.\varphi}$, we initially have
$\sem{\mu^0Z.\varphi} = \emptyset \subseteq\sem{\mu Z.\varphi}$, then by the monotonicity of $\varphi$,
given $\sem{\mu^iZ.\varphi}\subseteq\sem{\mu Z.\varphi}$, we prove $\sem{\mu^{i+1}Z.\varphi}\subseteq\sem{\mu Z.\varphi}$
by applying $\varphi$ on both sides of $\subseteq$. Therefore, $\sem{\mu^iZ.\varphi}\subseteq\sem{\mu Z.\varphi}$ for all $i\in\Nat$, thus
$\sem{\bigvee_{i\in\Nat}\mu^iZ.\varphi}\subseteq\sem{\mu Z.\varphi}$. To show $\sem{\mu Z.\varphi}\subseteq \sem{\mu^\omega Z.\varphi}$,
it is straightforward to see that $\mu^\omega Z.\varphi$ is a prefixpoint, therefore it contains $\mu Z.\varphi$, the intersection of all prefixpoints.

From Lemma~\ref{lem:approximants}, and by the soundness of PA-simulation to $\LP$ (Theorem~\ref{thm:soundness}),
we get the the soundness of PA-simulation to the logic $\LM$, as required.



\smallskip\noindent
{\bf Expressiveness of PAMu.}
There exist game-based extensions of probabilistic temporal logics, such as the logic PAMC~\cite{Song2019} that extends the
Alternating-time Mu-Calculus~\cite{AHK02}, and PATL~\cite{CL07} that extends PCTL~\cite{HJ94}. The semantics of both logics are
sets of states, rather than sets of distributions. It has also been shown in~\cite{Song2019} that PAMC and PATL are incomparable on
probabilistic game structures, based on a result showing that PCTL and P$\mu$TL are incomparable on Markov chains~\cite{Liu2015}.
Here we make a short comparison between PAMu and those logics.

Distribution formulas of PAMu cannot be expressed by state-based logics. 
For example, the formula $\strat{\pone}[\frac{1}{2},p]\oplus[\frac{1}{2},q]$,
expressing that player $\pone$ has a strategy to enforce 
in the next move a distribution which has half of its weight satisfying $p$ and the other half satisfying $q$,
cannot be expressed by PATL or PAMC.
As the latter two logics have probability values bundled with strategy modalities, a formula such as
$\strat{\pone}^{\geq\frac{1}{2}}p\wedge\strat{\pone}^{\geq\frac{1}{2}}q$ denotes that player $\pone$ has a strategy to enforce
$p$ with at least probability $\frac{1}{2}$ in the next step and player $\pone$ also has a \emph{possibly different} strategy to enforce
$q$ with at least probability $\frac{1}{2}$ in the next step. However, the resulting states (or distributions)
that satisfy $p$ and $q$ may overlap.

\begin{figure}[!t]
\centering
\includegraphics[height=2.0cm]{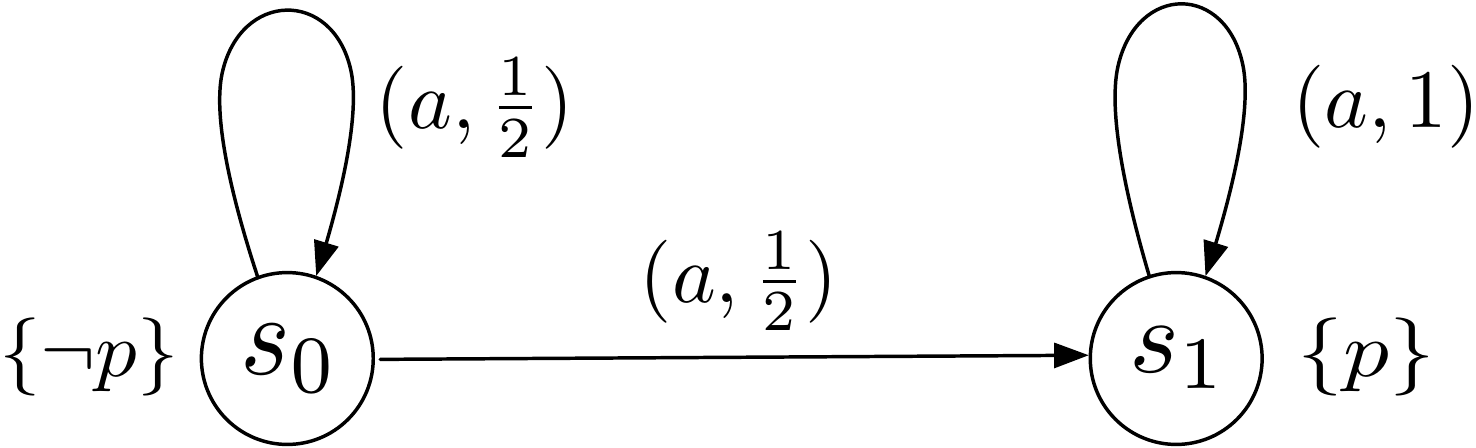}
\caption{An example for $\strat{\pone}^{\geq 1}\eventually p$.}\label{fig:example-PATL}
\end{figure}


\commentout{
The formula expressible by PATL that is not expressible by PAMC or P$\mu$TL (presented in~\cite{Liu2015}) is also not expressible by PAMu.
For example, the}
The PATL formula $\strat{\pone}^{\geq 1}\eventually p$ is not expressible by PAMu. 
Given the PGS in Figure~\ref{fig:example-PATL} where player $\pone$ has action set $\set{a}$ and player $\ptwo$ has action set $\emptyset$.
Then it is straightforward to see 
both $s_0$ and $s_1$ satisfies $\strat{\pone}^{\geq 1}\eventually p$.
The closest formula in PAMu is $\mu Z.p\vee\strat{\pone}Z$, but
$s_0\not\models\mu Z.p\vee\strat{\pone}Z$. 
More precisely, $s_0\models\mu Z.([\alpha,p]\oplus[1-\alpha,\top])\vee\strat{\pone}Z$
for all $0\leq\alpha<1$. Intuitively, the semantics of the least fixpoint operator in PAMu only track finite number of probabilistic transitions,
as starting from $s_0$, player $\pone$ can only reach distributions that satisfy
$p$ with probability strictly less than $1$ with finite number of steps.\footnote{
Intuitively, $\pd{s_0}\LT{a}[\frac{1}{2},s_0]\oplus[\frac{1}{2},s_1]\LT{a}[\frac{1}{4},s_0]\oplus[\frac{3}{4},s_1]\LT{a}\ldots\LT{a}[\frac{1}{2^i},s_0]\oplus[1-\frac{1}{2^i},s_1]\ldots$.
We shall see that in finite number of transitions one never reaches $\pd{s_1}$ from $\pd{s_0}$ with strict probability $1$.
However, such a restriction may be alleviated in practice, as
implemented in PRISM-game~\cite{KNPS18}, $\epsilon$-optimal strategies are synthesized for unbounded reachability properties.}
%
\commentout{
PAMC formulas in the form of $\strat{\pone}^{>\alpha}\varphi$ are not expressible by finite
conjunctions in PAMu. For example, $\strat{\pone}^{>\alpha}p$ may only be represented by (an infinite conjunction)
$\strat{\pone}\bigvee_{\beta>\alpha}[\beta,p]\oplus[1-\beta,\top]$ in PAMu. 
Moreover,}
PAMC formulas that contain the $\nostrat{\pone}^{\bowtie\alpha}$ modalities do not seem expressible in PAMu.
For instance, the PAMC formula $\nostrat{\pone}^{\geq\alpha}\varphi$~\footnote{$\nostrat{\pone}^{\geq\alpha}\varphi$ can be interpreted
as for all player $\pone$ strategies $\pi$, there exists player $\ptwo$ strategy $\sigma$, such that the combined effect
of $\pi$ and $\sigma$ enforces $\varphi$ with probability at least $\alpha$.} is semantically equivalent to $\neg\strat{\pone}^{<\alpha}\varphi$,
which is not expressible by PAMu as negation is only allowed at the propositional level in PAMu.
Nevertheless, the focus of PAMu is more on logic characterisation than expressiveness.

\begin{example}
The authors of~\cite{KNPS18} proposed a CSG variant of a futures market investor model~\cite{McIver2007},
which studies the interactions between an investor and a stock market. The
investor and the market take their decisions simultaneously in the CSG model,
and the authors show that this does not give any additional benefits to the investor
by analysing his maximum expected value over a fixed period of time.\footnote{
For details of the model, we refer to~\cite{McIver2007} and the website
{\small \url{https://www.prismmodelchecker.org}}.}
We take this example to show the expressiveness of PAMu.
For instance, the property ``it is always possible for the investor to cash in''
can be specified with two nested fixpoints as
\[
\nu X. (\mu Y. {\sf cashin} \lor \strat{\it investor} Y) \land \strat{\it investor} X).
\]
Another interesting property is to check whether the investor has a strategy to ensure
a good chance for him to make a profit.
This can be formulated in PAMu with $\frac{1}{2}<\alpha\leq 1$, as
\[
\mu Z. ({\sf cashin} \land [\alpha,{\sf profit}]\oplus[ 1-\alpha, \top]) \lor \strat{\it investor} Z
\]

\end{example}
\section{Related Work}
\label{sec:related}


Segala and Lynch~\cite{SL95} 
introduce a
probabilistic simulation relation
which
preserves probabilistic computation
tree logic (PCTL) formulas without negation and existential quantification.
Segala introduces the notion of probabilistic forward simulation,
which relates states to probability distributions over states
and is sound and complete for trace distribution precongruence~\cite{Seg95a,LSV07}.
Parma and Segala~\cite{PS07} 
use a probabilistic extension of the Hennessy-Milner logic
which allows countable conjunction and admits a new operator $[\phi]_p$
-- a distribution satisfies $[\phi]_p$ if the probability on the set of states satisfying $\phi$ is at least $p$,
with a sound and complete logic characterisation.
Hermanns et al.~\cite{HPSWZ11} further extend this result for image-infinite probabilistic automata.
Deng et al.~\cite{DGHMZ07,DG10} 
introduce a few probabilistic operators to derive a probabilistic modal mu-calculus (pMu).
A fragment of pMu 
is proved to characterise (strong) probabilistic simulation in finite-state probabilistic automata.

Alur, Henzinger and Kupferman~\cite{AHK02} define alternating-time temporal logic (ATL)
to generalise CTL for game structures by requiring each path quantifier to be parameterised by a set of agents.
GS are more general than LTS,
in the sense that they allow both collaborative and adversarial behaviours of
individual agents in a system, and ATL can be used to express properties like
``a set of agents can enforce a specific outcome of the system''.
The alternating simulation,
which is 
a natural game-theoretic interpretation of the classical simulation in (deterministic) multi-player games,
is introduced in~\cite{AHKV98}. 
Logic characterisation of this relation concentrates on
a subset of ATL$^\star$ formulas where negations are only allowed at propositional level
and all path quantifiers are parameterised by a predefined set of agents. 

Game structures deal well with systems in which the players
execute a \emph{pure strategy}, i.e., a strategy in which the moves
are chosen deterministically. However, 
\emph{mixed strategies}, which are formed by 
combining pure strategies, are necessary for a player to achieve \emph{optimal} rewards~\cite{NM47}.
Zhang and Pang~\cite{ZP10} extend the notion of game structures
to probabilistic game structures (PGS)
and introduce 
notions of simulation that are sound
for a fragment of probabilistic alternating-time
temporal logic (\PATL), a probabilistic extension of~ATL.

Fixpoint logics for sets of states in Markov chains and PGS have been studied more recently in~\cite{Liu2015,Song2019},
and a short comparison is given in  Section~\ref{sec:pamu}.

Metric-based simulation on game structures have been studied by de Alfaro et al~\cite{dAMRS08} regarding the
probability of winning games whose goals are expressed in quantitative $\mu$-calculus (qMu)~\cite{McIver2007}. Two
states are equivalent if the players can win the same games with the same probability from both states, 
and 
\emph{similarity} among states can thus be measured.
Algorithmic verification complexities are further
studied for MDP and turn-based games~\cite{CdAMR10}.

More recently, algorithmic verification of turn-based and concurrent games have been implemented as an
extension of PRISM~\cite{Kwiatkowska2018,KNPS18}. The properties can be specified as state formulas,
path formulas and reward formulas. The verification procedure requires solving matrix games for concurrent game structures,
and it applies value iteration algorithms to approach the goal (similar to~\cite{AM04,dAMRS08}).
For unbounded properties, the synthesised strategy
is memoryless (but only $\epsilon$-optimal strategies). Finite-memory strategies are synthesised for bounded properties.



\section{Conclusions and Future Work}
\label{sec:concl}

In this work, we have presented sound and complete modal characterisations of PA-simulation for concurrent games
by introducing two new logics $\LP$ and PAMu (with fixpoints).
Both logics incorporate nondeterministic and probabilisitic features and
express the ability of the players to enforce a property in current state. 
In the future, we aim to 
study verification complexities for these two logics. 


\bibliography{games}

\end{document}